\documentclass[letterpaper, 10 pt, conference]{ieeeconf}  
\IEEEoverridecommandlockouts
\usepackage{balance}
\usepackage{color}
\usepackage{caption}
\usepackage{enumerate}
\usepackage{cite}
\usepackage{empheq}
\usepackage{mathrsfs}

\usepackage{mathtools}

\usepackage{tikz}
\usepackage{circuitikz}

\usepackage[font=small,labelfont=bf]{caption}

\usepackage{subfig}


\makeatletter
\pgfcircdeclarebipole{}{\ctikzvalof{bipoles/interr/height 2}}{spst}{\ctikzvalof{bipoles/interr/height}}{\ctikzvalof{bipoles/interr/width}}{

    \pgfsetlinewidth{\pgfkeysvalueof{/tikz/circuitikz/bipoles/thickness}\pgfstartlinewidth}

    \pgfpathmoveto{\pgfpoint{\pgf@circ@res@left}{0pt}}
    \pgfpathlineto{\pgfpoint{.6\pgf@circ@res@right}{\pgf@circ@res@up}}
    \pgfusepath{draw}   
}

\def\pgf@circ@spst@path#1{\pgf@circ@bipole@path{spst}{#1}}
\tikzset{switch/.style = {\circuitikzbasekey, /tikz/to path=\pgf@circ@spst@path, l=#1}}
\tikzset{spst/.style = {switch = #1}}
\makeatother

\makeatletter
\let\proof\@undefined                        
\let\endproof\@undefined                  
\makeatother
\usepackage{graphicx,amssymb,amstext,amsmath,amsthm}

\usepackage[bookmarks=true]{hyperref}
\usepackage{algorithm,algorithmicx,algpseudocode}
\algnewcommand{\algorithmicgoto}{\textbf{go to}}%
\algnewcommand{\Goto}[1]{\algorithmicgoto~\ref{#1}}%
\algnewcommand{\LineComment}[1]{\Statex \(\triangleright\) #1}
\algnewcommand{\LineCommentN}[1]{\Statex \hspace{1cm}\(\triangleright\) #1}

\usepackage{multirow}
\usepackage{stfloats}

\newtheorem{prop}{Proposition} 

\newtheorem{thm}{Theorem}
	\newtheorem{assumption}{Assumption}
\newtheorem{lem}{Lemma}
\newtheorem{defn}{Definition}

\newtheorem{problem}{Problem}

\usepackage{setspace}

\let\oldbibliography\thebibliography
\renewcommand{\thebibliography}[1]{%
  \oldbibliography{#1}%
}

\newcommand{\yong}[1]{{\color{black} #1}}
\newcommand{\moh}[1]{{\color{black} #1}}
\newcommand{\mo}[1]{{\color{black} #1}}
\newcommand{\fa}[1]{{\color{black} #1}}
\newcommand{\mk}[1]{{\color{black} #1}}
\newcommand{\mkr}[1]{{\color{black} #1}}

\begin{document}

\title{\LARGE \bf Interval Observer Synthesis for Locally Lipschitz Nonlinear Dynamical Systems via Mixed-Monotone Decompositions} 

\author{%
Mohammad Khajenejad, Fatima Shoaib and Sze Zheng Yong\\
\thanks{
M. Khajenejad, Fatima Shoaib and S.Z. Yong are with the School for Engineering of Matter, Transport and Energy, Arizona State University, Tempe, AZ, USA. (e-mail: \{mkhajene, fshoaib, szyong\}@asu.edu.)}
\thanks{This work is partially supported by National Science Foundation grants CNS-1932066 and CNS-1943545.}
}

\maketitle
\thispagestyle{empty}
\pagestyle{empty}

\begin{abstract}
 This paper proposes a novel unified interval-valued observer synthesis approach for locally Lipschitz nonlinear continuous-time (CT) and discrete-time (DT) systems with nonlinear observations. A key feature of our proposed observer, which is derived using \emph{mixed-monotone decompositions}, is that it is
 \emph{correct} by construction (i.e., the true state trajectory of the system is \emph{framed} by the states of the observer) without the need for imposing additional constraints and assumptions such as global Lipschitz continuity or contraction, 
 as is done in existing approaches in the literature.  
 Furthermore, we derive sufficient conditions for designing stabilizing observer gains in the form of Linear Matrix Inequalities (LMIs).
 Finally, we \mk{compare the performance of our} observer design \mk{with some benchmark CT and DT observers in the literature}. 
 \end{abstract}

\section{Introduction} 

\mk{K}nowledge of system states is essential in almost all engineering applications, including fault detection, system identification and monitoring. However, in many realistic cases, system states are not fully measurable/measured and further, the sensor measurements may be limited or inaccurate. Thus, state observers have been designed to estimate system's states based on system dynamics and noisy/uncertain observations. When dealing with systems with set-valued (i.e., distribution-free) uncertainties, 
%
\emph{interval observers} have become increasingly popular due to their simple principles and computational efficiency \cite{wang2015interval,mazenc2013robust}.

The design of interval observers (a particular form of set-valued observers) has been extensively investigated in the literature for various classes of dynamical systems such as linear time-invariant (LTI) systems \cite{mazenc2011interval}, linear parameter varying (LPV)/quasi-LPV systems \cite{wang2015interval,chebotarev2015interval}, monotone/cooperative dynamics \cite{moisan2007near,raissi2010interval}, Metzler systems \cite{mazenc2013robust} and mixed-monotone dynamics \cite{tahir2021synthesis,khajenejad2021intervalACC,khajenejad2020simultaneousCDC}. To obtain cooperative observer error dynamics, the design of interval observers has either directly relied on monotone systems theory \cite{farina2000positive}, 
 or relatively restrictive assumptions about the existence of certain system properties were imposed to guarantee the {applicability} 
of the proposed approaches. 
\mk{However,} even for linear systems, it is not easy nor guaranteed to synthesize the framer gain \mk{to} 
satisfy correctness and stability at the same time \cite{chambon2016overview}. 
This difficulty to obtain such properties was relaxed for certain classes of systems, by applying time-invariant/varying state transformations \cite{tahir2021synthesis,mazenc2011interval}, transformation to a positive system before designing an observer \cite{cacace2014new} (only applicable to linear systems) or leveraging interval arithmetic or M\"{u}ller's theorem-based approaches \cite{kieffer2006guaranteed}. 

In the context of nonlinear systems, an interval observer design has been proposed in \cite{efimov2013interval} for a class of continuous-time nonlinear systems by
leveraging \emph{bounding functions}, but no necessary and/or sufficient conditions for the existence of bounding functions or how to compute them have been discussed. Moreover, to conclude stability, restrictive assumptions on the nonlinear dynamics have been imposed. {On the other hand,} the authors in \cite{tahir2021synthesis} applied bounding/mixed-monotone decomposition functions to design interval state estimation for nonlinear discrete-time dynamics, where to guarantee positivity of the error dynamics (i.e., the correctness property), conservative additive terms were added to the error dynamics. Moreover, to best of our unnderstanding, the required conditions to guarantee that the computed bounding functions are decomposition functions were not included in the resulting Linear Matrix Inequalities (LMIs). 
On the other hand, our previous work \cite{khajenejad2021intervalACC,khajenejad2020simultaneousCDC} designed interval observer for globally Lipschitz mixed-monotone nonlinear discrete-time systems, where the stability of the proposed observer relied on some sufficient structural system properties. 


 In this paper, we introduce a novel method for synthesizing interval observers using mixed-monotone decompositions for 
 locally Lipschitz nonlinear CT and DT systems with nonlinear observation functions. The main feature and advantage of our proposed observer is that it is \emph{correct} by construction.
 In particular, by leveraging \emph{remainder-form mixed-monotone decomposition functions}, we show that the true state trajectory of the system is guaranteed to frame the true states of the observer by construction. In other words, the observer error system is by design positive (for DT systems) or cooperative (for CT systems) without the need for additional assumptions, e.g., global Lipschitz continuity and contraction. 
 Moreover, we derive sufficient conditions in the form of LMIs that ensure that our proposed correct-by-construction interval observer is also stable, and they can be utilized to design stabilizing observer gains via semi-definite programming. 
 Finally, our unified framework is the first to address the problem of synthesizing interval observers for a very broad range of  locally Lipschitz CT and DT systems that are correct by construction.   
 \section{Preliminaries}
 
 {\emph{{Notation}.}} $\mathbb{R}^n,\mathbb{R}^{n  \times p},\mathbb{D}_n,\mathbb{N},\mathbb{N}_n$ denote the $n$-dimensional Euclidean space and the sets of $n$ by $p$ matrices, $n$ by $n$ diagonal matrices, natural numbers and natural numbers up to $n$, respectively, while $\mathbb{M}_n$ 
 denote\yong{s} the set of all $n$ by $n$ Metzler\footnote{A Metzler matrix is a square matrix in which all the off-diagonal components are nonnegative (equal to or greater than zero).} \mo{matrices}. 
 For $M \in \mathbb{R}^{n \times p}$, $M_{i,j}$ denotes $M$'s entry in the $i$'th row and the $j$'th column, $M^+\triangleq \max(M,\mathbf{0}_{n,p})$, $M^-=M^+-M$ and $|M|\triangleq M^++M^-$, where $\mathbf{0}_{n,p}$ is the zero matrix in $R^{n \times p}$, \yong{while $\textstyle{\mathrm{sgn}}(M) \in \mathbb{R}^{n \times p}$ is the element-wise sign of $M$ with $\textstyle{\mathrm{sgn}}(M_{i,j})=1$ if $M_{i,j} \geq 0$ and $\textstyle{\mathrm{sgn}}(M_{i,j})=-1$, otherwise.} 
 Further, if $p=n$, $M^\text{d}$ denotes a diagonal matrix whose diagonal coincides with the diagonal of $M$, $M^\text{nd} \triangleq M-M^\text{d}$ and $M^{\text{m}} \triangleq M^\text{d}+|M^\text{nd}|$, \mk{while $M \succ 0$ and $M \prec 0$ (or $M \succeq 0$ and $M \preceq 0$) denote that $M$ is positive and negative (semi-)definite, respectively}.

Next, we introduce some useful definitions and results.
\begin{defn}[Interval, Maximal and Minimal Elements, Interval Width]\label{defn:interval}
{An (multi-dimensional) interval {$\mathcal{I} \triangleq [\underline{s},\overline{s}]  \subset 
\mathbb{R}^n$} is the set of all real vectors $x \in \mathbb{R}^n$ that satisfies $\underline{s} \le x \le \overline{s}$, where $\underline{s}$, $\overline{s}$ and $\|\overline{s}-\underline{s}\|\mk{_{\infty}\triangleq \max_{i \in \{1,\cdots,n\}}s_i}$ are called minimal vector, maximal vector and \mk{interval} width of $\mathcal{I}$, respectively}. \yong{An interval matrix can be defined similarly.} 
\end{defn}
\begin{prop}\cite[Lemma 1]{efimov2013interval}\label{prop:bounding}
Let $A \in \mathbb{R}^{n \times p}$ and $\underline{x} \leq x \leq \overline{x} \in \mathbb{R}^n$. Then\moh{,} $A^+\underline{x}-A^{-}\overline{x} \leq Ax \leq A^+\overline{x}-A^{-}\underline{x}$. As a corollary, if $A$ is non-negative, $A\underline{x} \leq Ax \leq A\overline{x}$. 
\end{prop}

\begin{defn}[Jacobian Sign-Stability] \label{defn:JSS}
A mapping $f :\mathcal{X} \subset \mathbb{R}^{n} \to  \mathbb{R}^{p}$ is (generalized) Jacobian sign-stable (JSS), if its (generalized) Jacobian matrix entries \fa{do} not change signs on its domain, i.e., if either of the following hold: 
\begin{align*}
&\forall x \in \mathcal{X}, \forall i \in \mathbb{N}_p,\forall j \in \mathbb{N}_n , J_f(x)_{i,j} \geq 0 \ \text{(positive JSS)}\\  
&\forall x \in \mathcal{X}, \forall i \in \mathbb{N}_p,\forall j \in \mathbb{N}_n , J_f(x)_{i,j} \leq 0 \  \text{(negative JSS)},
\end{align*} 
where $J_f(x)$ denotes the Jacobian matrix of $f$ at $x \in \mathcal{X}$. 
\end{defn}
\begin{prop}[Jacobian Sign-Stable Decomposition]\label{prop:JSS_decomp}
Let $f :\mathcal{X} \subset \mathbb{R}^{n} \to  \mathbb{R}^{p}$ and suppose $\forall x \in \mathcal{X}, J_f(x) \in [\underline{J}_f,\overline{J}_f]$, where $\underline{J}_f,\overline{J}_f$ are known matrices in $\mathbb{R}^{p \times n}$. Then, $f$ can be decomposed into a (remainder) affine mapping $Hx$ and a JSS mapping $\mu (\cdot)$, in an additive form: 
\begin{align}\label{eq:JSS_decomp}
\forall x \in \mathcal{X},f(x)=\mu(x)+Hx,
\end{align}
 where $H$ is a matrix in $\mathbb{R}^{p \times n}$, that satisfies the following 
 \begin{align}\label{eq:H_decomp}
 \forall (i,j) \in \mathbb{N}_p \times \mathbb{N}_n, H_{i,j}=(\overline{J}_f)_{i,j} \ \lor \ H_{i,j}=(\underline{J}_f)_{i,j}.    
 \end{align}
\end{prop}
\begin{proof}
Let us define $\mu(x) \triangleq f(x)-Hx$, where $H$ is given in \eqref{eq:H_decomp}. Then, it follows from \eqref{eq:JSS_decomp} that $\forall x \in \mathcal{X},\forall (i,j) \in \mathbb{N}_p \times \mathbb{N}_n,(J_{\mu})_{i,j}(x)=(J_{f})_{i,j}(x)-H_{i,j}$. From this and given the fact that $\forall x \in \mathcal{X}, J_f(x) \in [\underline{J}_f,\overline{J}_f]$, we obtain $(\underline{J}_{f})_{i,j}-H_{i,j}\leq(J_{\mu})_{i,j}(x)\leq(\overline{J}_{f})_{i,j}-H_{i,j}$. Now, i) if \yong{we have} the case $H_{i,j}=\overline{f}_{i,j}$ in \eqref{eq:H_decomp}, 
the right inequality implies $(J_{\mu})_{i,j}(x) \leq \overline{J}_f-\overline{J}_f=0$, i.e., $\mu_i$ is monotonically increasing in the $j$'th dimension. On the other hand, ii) if \yong{we have} the case $H_{i,j}=\underline{f}_{i,j}$ in \eqref{eq:H_decomp}, 
it follows from the left inequality that $(J_{\mu})_{i,j}(x) \geq \underline{J}_f-\underline{J}_f=0$, i.e., $\mu_i$ is monotonically decreasing in the $j$'th dimension. Hence, $\mu$ is a JSS mapping, by \yong{either i) or ii).} 
\end{proof}
\begin{defn}[Mixed-Monotonicity {and} Decomposition Functions] \cite[Definition 1]{abate2020tight},\cite[Definition 4]{yang2019sufficient} \label{defn:dec_func}
Consider the dynamical system {with initial state $x_0 \in \mathcal{X}_0 \triangleq [\underline{x}_0,\overline{x}_0]$:}
\begin{align}\label{eq:mix_mon_def}
x_t^+= g(x_t),
\end{align}
where $x_t^+ \triangleq x_{t+1}$ if \eqref{eq:mix_mon_def} is a DT {system} and $x_t^+ \triangleq \dot{x}_t$ if \eqref{eq:mix_mon_def} is a CT system. Moreover, ${g}:\mathcal{X} \subset \mathbb{R}^n \to \mathbb{R}^{p}$ is the vector field 
with state $x_t \in \mathcal{X}$ \yong{as its domain}. 

Suppose \eqref{eq:mix_mon_def} is a DT system. 
Then, a mapping $g_d:\mathcal{X}\times \mathcal{X} \to \mathbb{R}^{p}$ is 
a {DT mixed-monotone} decomposition function with respect to $g$, 
if i) $g_d(x,x)=g(x)$, ii) $g_d$ is monotone increasing in its first 
argument, i.e., $\hat{x}\ge x \Rightarrow g_d(\hat{x},x') \geq g_d(x,x')$, and iii) {$g_d$ is monotone decreasing in its second argument, i.e., $\hat{x}\ge x \Rightarrow g_d(x',\hat{x}) \leq g_d(x',x).$}

\yong{On the other hand,} if \eqref{eq:mix_mon_def} is a CT system, 
a mapping $g_d:\mathcal{X}\times \mathcal{X} \to \mathbb{R}^{p}$ is a {CT mixed-monotone} decomposition function with respect to $g$, 
if i) $g_d(x,x)=g(x)$, ii) $g_d$ is monotone increasing in its first 
argument with respect to ``off-diagonal'' arguments, i.e., $\forall (i,j) \in \mathbb{N}_{p} \times \mathbb{N}_{n} \land i \ne j,\hat{x}_j\ge x_j, \hat{x}_i= x_i  \Rightarrow g_{d,i}(\hat{x},x') \geq g_{d,i}(x,x')$, and iii) {$g_d$ is monotone decreasing in its second argument, i.e., $\hat{x}\ge x \Rightarrow g_d(x',\hat{x}) \leq g_d(x',x).$}   
\end{defn}
\begin{defn}[Embedding {Systems}]\label{def:embedding}
For an $n$-dimensional 
system \eqref{eq:mix_mon_def} {with any decomposition functions ${g}_d(\cdot)$, its \emph{embedding system is the following} 
$2n$-dimensional system {with initial condition $\begin{bmatrix} \overline{x}_0^\top & \underline{x}_0^\top\end{bmatrix}^\top$}:}
\begin{align} \label{eq:embedding}
\begin{bmatrix}{(\overline{x}}_t^+)^\top & ({\underline{x}}_t^+)^\top \end{bmatrix}^\top=\begin{bmatrix} {g}^\top_d(\overline{x}_t,\underline{x}_t) & {g}^\top_d(\underline{x}_t,\overline{x}_t)\end{bmatrix}^\top.
\end{align}
\end{defn}

\begin{prop}[State Framer Property]\label{cor:embedding} \cite[Proposition 3]{khajenejad2021tight}
{Let system \eqref{eq:mix_mon_def} with initial state $x_0 \in \mathcal{X}_0 \triangleq  [\underline{x}_0,\overline{x}_0]$ be mixed-monotone with an 
 embedding system \eqref{eq:embedding} with respect to 
${g}_d$. Then, for all $t \in \yong{\mathbb{T}}$, $R^g(t,\mathcal{X}_0) \subset \mathcal{X}_t \triangleq [\underline{x}_t,\overline{x}_t]$, 
where $R^g(t,\mathcal{X}_0) \triangleq
\{\mu_g(t, x_0) \mid x_0 \in \mathcal{X}_0,  \forall t \in \yong{\mathbb{T}}\}$ is the reachable set at time $t$ of \eqref{eq:mix_mon_def} 
when initialized within $\mathcal{X}_0$, $\mu_g(t,x_0)$ is the true state trajectory function of system \eqref{eq:mix_mon_def} and $(\overline{x}_t,\underline{x}_t)$ is the solution to the embedding system \eqref{eq:embedding}, \yong{with $\mathbb{T} \in \mathbb{R}_{\ge 0}$ for CT systems and $\mathbb{T} \in \{0\} \cup \mathbb{N}$ for DT systems}. Consequently,
the system state trajectory $x_t=\mu_g(t,x_0)$ 
satisfies $\underline{x}_t \le x_t \le \overline{x}_t, {\forall t \geq 0}$, i.e., \yong{it} is \emph{framed} by $\mathcal{X}_t \triangleq [\underline{x}_t,\overline{x}_t]$.}
\end{prop}

\begin{prop}[Tight and Tractable Decomposition Functions for JSS Mappings]\label{prop:tight_decomp}
Let $\mu:\mathcal{X} \subset \mathbb{R}^n \to \mathbb{R}^p$ be a JSS mapping on its domain. Then, it admits a tight decomposition function that has the following form for any \mo{ordered} \yong{$x_1, x_2 \in \mathcal{X}$:} 
\begin{align}\label{eq:JJ_decomp}
\forall i \in \mathbb{N}_p, \mu_{d,i}({x}_1,{x}_2)\hspace{-.1cm}=\hspace{-.1cm}\mu_i(D_i{x}_1\hspace{-.1cm}-\hspace{-.1cm}(I_n\hspace{-.1cm}-\hspace{-.1cm}D_i){x}_2), 
\end{align}
where $D_i \in \mathbb{D}_n$ is a binary diagonal matrix determined by which vertex of the interval $[{x}_2,{x}_1]$ \mo{or $[x_1,x_2]$} that maximizes \yong{(if $x_2 \leq x_1$) or minimizes (if $x_2 > x_1$)} the JSS function $\mu_i(\cdot)$ and can be computed as follows:
\begin{align}\label{eq:Dj}
D_i=\textstyle{\mathrm{diag}}(\max(\textstyle{\mathrm{sgn}}(\yong{\overline{J}_{\mu,i}}),\mathbf{0}_{1,n})).
\end{align}
\end{prop}
\begin{proof}
First, by \cite[corollary 2]{khajenejad2021tight}, the JSS function $\mu(\cdot)$ admits a tight decomposition function. \yong{Then,} if $x_2 \leq x_1$, \yong{$\lambda(\cdot)$} can be computed by solving the following nonlinear optimization (cf. \cite[(7)]{khajenejad2021tight}): $\mu_{d,i}(x_1,x_2)=\max_{z \in [x_2,x_1]}\mu_i(z)$. Suppose that $z^* \in [x_2,x_1]$ is the maximizer  and consider the case that i) $\mu_i$ is positive JSS in dimension $j$. Obviously $z^*_j=x_{1,j}$. On the other hand, $(J_\mu)_{i,j} \geq 0$, and so, \yong{$\max(\textstyle{\mathrm{sgn}}((J_\mu)_{i,j}),0)=1$}. Hence, by \eqref{eq:Dj}, $(D_i)_{j,j}=1$, and therefore in \eqref{eq:JJ_decomp}, $(D_i{x}_1-(I_n-D_i){x}_2)_j=x_{1,j}$, which is consistent with what we obtained for the maximizer's $j$'th entry, i.e., $z^*_j=x_{1,j}$. Similar reasoning shows that such consistency \yong{also} holds in the negative JSS case as well \yong{as when finding the minimizer of $\mu_i(\cdot)$ if $x_2>x_1$}. 

Further, 
since $\mu(\cdot)$ is JSS, $J_{\mu,i}$ does not change signs 
and hence, \yong{$\max(\textstyle{\mathrm{sgn}}(J_{\mu,i}),0)$} is well-defined \yong{and we can equivalently use $\max(\textstyle{\mathrm{sgn}}(J_{\mu,j}),0) = \max(\textstyle{\mathrm{sgn}}(\overline{J}_{\mu,j}),0)$}. 
\end{proof}

\section{Problem Formulation} \label{sec:Problem}
\noindent\textbf{\emph{System Assumptions.}} 
Consider the following nonlinear  continuous-time (CT) or discrete-time (DT) system:  
\begin{align} \label{eq:system}
\begin{array}{ll}
\mathcal{G}: \begin{cases} {x}_t^+ = \hat{f}(x_t,u_t) \triangleq f(x_t),   \\
                                              y_t = \hat{h}(x_t,u_t) \triangleq h(x_t), 
                                              \end{cases}, x_t\in \mathcal{X}, t \in \yong{\mathbb{T}},
\end{array}\hspace{-0.2cm}
\end{align}
where $x_t^+=\dot{x}_t, \yong{\mathbb{T}} = \mathbb{R}_{\ge 0}$ if $\mathcal{G}$ is a CT and $x_t^+=x_{t+1}, \yong{\mathbb{T}}= \{0\}\cup \mathbb{N}$, if $\mathcal{G}$ is a DT system. Moreover, $x_t \in \mathcal{X} \subset \mathbb{R}^n$, $u_t \in \mathbb{R}^s$ and $y_t \in \mathbb{R}^l$ are continuous state, known (control) input and output (measurement) signals. Furthermore, $\hat{f}:\mathbb{R}^n \times \mathbb{R}^s \to \mathbb{R}^n$ and $\hat{h}:\mathbb{R}^n \times \mathbb{R}^s \to \mathbb{R}^l$ are nonlinear state vector field and observation/constraint \yong{functions/}mappings, respectively, from which, 
the \yong{functions/}mappings $f:\mathbb{R}^n \to \mathbb{R}^n$ and $g:\mathbb{R}^n \to \mathbb{R}^l$ are well-defined \yong{since the input signal $u_t$ is known}.  
We are interested in estimating the trajectories
of the plant $\mathcal{G}$ in \eqref{eq:system}, when they are initialized in a given interval
$\mathcal{X}_0 \subset \mathcal{X} \subset \mathbb{R}^n$.
We 
also assume the following: 
\begin{assumption} \label{ass:initial_interval}
 The initial state $x_0$ satisfies $x_0 \in \mathcal{X}_0 = [ \underline{x}_0,\overline{x}_0]$, where $\underline{x}_0$ and $\overline{x}_0$ {are} known initial state bounds. 
 \end{assumption}
 \begin{assumption}\label{ass:mixed_monotonicity}
 The 
 \yong{mappings} $f(\cdot)$ and $h(\cdot)$ are known, \mk{differentiable}, locally Lipschitz\mk{\footnote{\mk{Both assumptions of locally Lipschitz continuity and differentiability can be relaxed to a much weaker continuity assumption (cf. \cite{khajenejad2021tight} for more details), and these assumptions are mainly made for ease of exposition.}}} and mixed-monotone in their domain with priori known upper and lower bounds for their Jacobian matrices, $\overline{J}_{f},\underline{J}_{f} \in \mathbb{R}^{n \times n}$ and $\overline{J}_{h},\underline{J}_{h} \in \mathbb{R}^{l \times n}$, respectively.
\end{assumption}
\begin{assumption}\label{ass:known_input_output}
 The values of the input $u_t$ and output/measurement $y_t$ signals are known at all times. 
\end{assumption}
Further, we formally define the notions of \emph{framers}, \emph{correctness} and \emph{stability} that are used throughout the paper. 
\begin{defn}[Correct Interval \mk{Framers}]\label{defn:framers}
Suppose Assumptions \ref{ass:initial_interval}--\ref{ass:known_input_output} hold. Given the nonlinear plant \eqref{eq:system}, 
the mappings/signals $\overline{x},\underline{x}: \mo{\mathbb{T}} \to \mathbb{R}^n$ are called upper and lower framers for the states of System \eqref{eq:system}, if 
\begin{align}\label{eq:correctness}
\forall t \in \yong{\mathbb{T}}, \ \underline{x}_t \leq x_t \leq \overline{x}_t.
\end{align}
In other words, starting from the initial interval $\underline{x}_0 \leq x_0 \leq \overline{x}_0$, the true state of the system in \eqref{eq:system}, $x_t$, is guaranteed to evolve within the interval flow-pipe $[\underline{x}_t,\overline{x}_t]$, for all $t \in \yong{\mathbb{T}}$. Finally, any dynamical system whose states are correct framers for the states of the plant $\mathcal{G}$, i.e., any (tractable) algorithm that returns upper and lower framers for the states of plant $\mathcal{G}$ is called a \emph{correct} interval \mk{framer} for system \eqref{eq:system}. 
\end{defn}
\begin{defn}[\mk{Framer} Error]\label{defn:error}
Given 
state framers \mk{$\underline{x}_t \leq \overline{x}_t$}, $\varepsilon : \mo{\mathbb{T}} \to \mathbb{R}^n$, \mk{which denotes} the interval width \mk{of $[\underline{x}_t,\overline{x}_t]$ (cf. Definition \ref{defn:interval})},  
 is called the \mk{framer} error. It can be easily verified that 
correctness (cf. Definition \ref{defn:framers}) implies that $\varepsilon_t \geq 0, \forall t \in \mo{\mathbb{T}}.$  
\end{defn}
\begin{defn}[Stability and \mk{Interval Observer}]\label{defn:stability}
An interval \mk{framer} is stable, if the \mk{framer} error (cf. Definition \ref{defn:error}) asymptotically converges to zero, i.e., 
$\lim_{t \to \infty} \|\varepsilon _t\|=0$. \mk{A stable interval framer is called an interval observer.}
\end{defn}

The 
observer design problem 
 can be stated as follows:
\begin{problem}\label{prob:SISIO}
Given the nonlinear system in \eqref{eq:system}, as well as Assumptions \ref{ass:initial_interval}--\ref{ass:known_input_output}, 
synthesize an interval observer, \mk{i.e., a} correct and stable \mk{framer} (cf. Definitions \ref{defn:framers} and \ref{defn:stability}). 
\end{problem}

\section{Proposed Interval Observer} \label{sec:observer}
\subsection{Interval Observer Design} \label{sec:obsv}
Given the nonlinear plant $\mathcal{G}$, in order to address Problem \ref{prob:SISIO}, we propose an interval observer (cf. Definition \ref{defn:framers}) for $\mathcal{G}$ through the following dynamical system:
\begin{align}\label{eq:observer}
\hat{\mathcal{G}}:\begin{cases}\overline{x}_t^+\hspace{-.4cm}&=(A\hspace{-.1cm}-\hspace{-.1cm}LC)^\uparrow \overline{x}_t-(A\hspace{-.1cm}-\hspace{-.1cm}LC)^\downarrow \underline{x}_t\hspace{-.1cm}+\hspace{-.1cm}Ly_t \\
&+\phi_d(\overline{x}_t,\underline{x}_t)\hspace{-.1cm}-\hspace{-.1cm}L^+\psi_d(\underline{x}_t,\overline{x}_t)\hspace{-.1cm}+\hspace{-.1cm}L^-\psi_d(\overline{x}_t,\underline{x}_t) \\
\underline{x}_t^+\hspace{-.4cm}&=(A\hspace{-.1cm}-\hspace{-.1cm}LC)^\uparrow \underline{x}_t-(A\hspace{-.1cm}-\hspace{-.1cm}LC)^\downarrow \overline{x}_t\hspace{-.1cm}+\hspace{-.1cm}Ly_t \\
&+\phi_d(\underline{x}_t,\overline{x}_t)\hspace{-.1cm}-\hspace{-.1cm}L^+\psi_d(\overline{x}_t,\underline{x}_t)\hspace{-.1cm}+\hspace{-.1cm}L^-\psi_d(\underline{x}_t,\overline{x}_t) 
 \end{cases},
\end{align}
where if $\mathcal{G}$ is a CT system, then
\begin{align}\label{eq:T_CT}
\begin{array}{rl}
\overline{x}_t^+\hspace{-.2cm}&\triangleq \dot{\overline{x}_t},(A-LC)^\uparrow \triangleq (A-LC)^\text{d}+(A-LC)^{\text{nd}+},\\
\underline{x}_t^+\hspace{-.2cm}&\triangleq \dot{\underline{x}_t},(A-LC)^\downarrow \triangleq (A-LC)^{\text{nd}-},
\end{array}
\end{align}
and if $\mathcal{G}$ is a DT system, then
\begin{align}\label{eq:T_DT}
\begin{array}{rl}
\overline{x}_t^+&\triangleq \overline{x}_{t+1},(A-LC)^\uparrow \triangleq (A-LC)^+,\\
\underline{x}_t^+&\triangleq \underline{x}_{t+1},(A-LC)^\downarrow \triangleq (A-LC)^{-}.
\end{array}
\end{align}
Moreover, $A \in \mathbb{R}^{n \times n}$ and $C \in \mathbb{R}^{l \times n}$ are chosen such that the following decompositions hold (cf. Definition \ref{defn:JSS} and Proposition \ref{prop:JSS_decomp}):
\begin{align} \label{eq:JSS_decom}
\hspace{-.2cm}\forall x \in \mathcal{X}: \begin{cases}f(x)=Ax+\phi(x) \\ h(x)=Cx+\psi(x) \end{cases} \hspace{-.4cm} s.t. \ \phi,\psi \ \text{are JSS in} \ \mathcal{X}.
\end{align}  
Furthermore, $\phi_d:\mathbb{R}^n \times \mathbb{R}^n \to \mathbb{R}^n$ and $\psi_d:\mathbb{R}^n \times \mathbb{R}^n \to \mathbb{R}^l$ are tight mixed-monotone decomposition functions of $\phi$ and $\psi$, respectively (cf. Definition \ref{defn:dec_func} and Propositions \ref{cor:embedding}--\ref{prop:tight_decomp}). Finally, $L \in \mathbb{R}^{n \times l}$ is the observer gain matrix, designed via Theorem \ref{thm:stability}, such that the proposed observer $\hat{\mathcal{G}}$ possesses the desired properties discussed in the following subsections.
\subsection{Observer Correctness (Framer Property)}\vspace{-0.05cm}
Our strategy is to design a \emph{correct by construction} interval observer for plant $\mathcal{G}$. To accomplish this goal, first, note that from \eqref{eq:observer} and \eqref{eq:JSS_decom}
we have $y_t-Cx_t-\psi(x_t)=0$, and so $L(y_t-Cx_t-\psi(x_t))=0$, for any $L \in \mathbb{R}^{n \times l}$. Adding this ``zero" term to the right hand side of \eqref{eq:system} \yong{and applying} 
\eqref{eq:JSS_decom}
yield the following equivalent system to $\mathcal{G}$:  
\begin{align} \label{eq:eqiv_sys}
x_t^+=(A-LC)x_t+Ly_t+\phi(x_t)-L\psi(x_t).
\end{align}
From now on, we are interested in computing embedding systems, in the sense of Definition \ref{def:embedding}, for the system in \eqref{eq:eqiv_sys}, so that by Proposition \ref{cor:embedding}, the state trajectories of \eqref{eq:eqiv_sys} are ``framed" by the state trajectories of the computed embedding system. To do so, we split the right hand side of \eqref{eq:eqiv_sys} (except for $Ly_t$ that is independent \yong{of the states}) into two constituent systems: the linear constituent $(A-LC)x_t$ and the nonlinear constituent, $\phi(x_t)-L\psi(x_t)$. Then, we compute embedding systems for each constituent, separately. Finally, we add the computed embedding systems to construct an embedding system for \eqref{eq:eqiv_sys}. We start with framing the linear constituent through the following lemma.
\begin{lem}[Linear Embedding]\label{lem:linear_bounding}
Consider a dynamical system ${\mathcal{G}}_{\ell}$ in the form of \eqref{eq:mix_mon_def}, with domain $\mathcal{X}$ and state equation ${f}_{\ell}(x_t)=(A-LC)x_t$. Then, a tight decomposition function (cf. Definition \ref{defn:dec_func}) for ${\mathcal{G}}_{\ell}$ can be computed as follows:
\begin{align} \label{eq:Lin_dec}
\tilde{f}_{\ell d}(x_1,x_2)=(A-LC)^\uparrow x_1-(A-LC)^\downarrow x_2,
\end{align}
where $(A-LC)^\uparrow $ and $(A-LC)^\downarrow$ are given in \eqref{eq:T_CT} and \eqref{eq:T_DT} for CT and DT systems, respectively. 
\end{lem}
\begin{proof} We start with the DT case, where $(A-LC)^\uparrow \triangleq (A-LC)^+,(A-LC)^\downarrow \triangleq(A-LC)^-$. It is easy to verify that $\tilde{\mo{f}}_{\ell d}$ is increasing in $x_1$ since $(A-LC)^+ \geq 0$, is decreasing in $x_2$ since $-(A-LC)^- \leq 0$, and $\tilde{\mo{f}}_{\ell d}(x,x)=((A-LC)^+-(A-LC)^-)x=(A-LC)x={\mo{f}}_{\ell}(x)$. Hence, $\tilde{\mo{f}}_{\ell d}$ is a DT decomposition function of $\tilde{\mo{f}}$. The proof for tightness goes through similar lines of the proof of \cite[Lemma 1]{khajenejad2021simultaneousECC}. As for the CT case, where $(A-LC)^\uparrow \triangleq (A-LC)^{\text{d}}+(A-LC)^{\text{nd}+},(A-LC)^\downarrow \triangleq(A-LC)^{\text{nd}-}$, the proof is similar to the one for the DT case, with the slight difference that in the CT case, we need increasing monotonicity of $\tilde{\mo{f}}_{\ell d}$ only in off-diagonal elements of $x_1$ (cf. Definition \ref{defn:dec_func}), which is guaranteed by non-negativity of $(A-LC)^{\text{nd}+}$.         
\end{proof}   
Next, we compute an embedding system for the nonlinear constituent system in \eqref{eq:eqiv_sys}, i.e., $\phi(x_t)-L\psi(x_t)$, as follows.
\begin{lem}[Nonlinear Embedding]\label{lem:nonlinear_bounding}
Consider a dynamical system ${\mathcal{G}}_{{\nu}}$ in the form of \eqref{eq:mix_mon_def}, with domain $\mathcal{X}$ and state equation ${f}_{\nu}(x_t)=\phi(x_t)-L\psi(x_t)$. Then, a decomposition function (cf. Definition \ref{defn:dec_func}) for ${\mathcal{G}}_{\nu}$ can be computed as follows:
\begin{align}\label{eq:Lin_dec}
\hspace{-.2cm}{f}_{\nu d}(x_1,x_2\hspace{-.05cm})\hspace{-.1cm}= \hspace{-.1cm}\phi_d({x}_1,{x}_2\hspace{-.05cm})\hspace{-.1cm}-\hspace{-.1cm}L^+\psi_d({x}_2,{x}_1\hspace{-.05cm})\hspace{-.1cm}+\hspace{-.1cm}L^-\psi_d({x}_1,{x}_2),
\end{align} 
where $\phi_d(\cdot,\cdot),\psi_d(\cdot,\cdot)$ are tight decomposition functions for the JSS mapping $\phi(\cdot),\psi(\cdot)$, computed via Proposition \ref{prop:tight_decomp}. 
\end{lem}
\begin{proof}
${f}_{\nu d}$ is increasing in $x_1$ since it is a summation of three increasing mappings in $x_1$, including $\phi_d(x_1,x_2)$ (a decomposition function that by construction is increasing in $x_1$), $-L^+\psi_d({x}_2,{x}_1)$ (a multiplication of the non-positive matrix $-L^+$ and the decomposition function $\psi_d({x}_2,{x}_1)$ which is decreasing on $x_1$ by construction) and $L^-\psi_d({x}_1,{x}_2)$ (a multiplication of the non-negative matrix $L^-$ and the decomposition function $\psi_d({x}_1,{x}_2)$ which is itself increasing on $x_1$ by construction). Similar reasoning shows that ${g}_{\nu d}$ is decreasing in $x_2$. Finally, ${f}_{\nu d}(x,x)=\phi_d({x},{x})-L^+\psi_d({x},{x})+L^-\psi_d({x},{x})=\phi(x)-L\psi(x)=f_{\nu}(x)$. 
\end{proof}
We conclude this subsection by combining the results in Lemmas \ref{lem:linear_bounding} and \ref{lem:nonlinear_bounding}, as well as Proposition \ref{cor:embedding}, that results in the following theorem on correctness of the proposed observer.
\begin{thm}[\mk{Correct Interval Framer}]\label{lem:correctness}
Consider the nonlinear plant $\mathcal{G}$ in \eqref{eq:system} and suppose Assumptions \ref{ass:initial_interval}--\ref{ass:known_input_output} hold. Then, the dynamical system $\hat{\mathcal{G}}$ in \eqref{eq:observer} constructs a correct interval \mk{framer} for the nonlinear plant $\mathcal{G}$. In other words, $\forall t \in \mo{\mathbb{T}}, \underline{x}_t \leq x_t \leq \overline{x}_t$, where $x_t$ and $[\overline{x}_t^\top \underline{x}_t^\top]^\top$ are the state vectors in $\mathcal{G}$ and $\hat{\mathcal{G}}$ at time $t \in \mo{\mathbb{T}}$, respectively. 
\end{thm}
\begin{proof}
It is straightforward to show that the summation of decomposition functions of constituent systems, is a decomposition function of the summation of the constituent systems. 
\yong{Combining this with} Lemmas \ref{lem:linear_bounding} and \ref{lem:nonlinear_bounding} implies that $f_d(x_1,x_2) \triangleq g_{\ell d}(x_1,x_2)+g_{\nu d}(x_1,x_2)=(A-LC)^\uparrow x_1-(A-LC)^\downarrow x_2+\phi_d({x}_1,{x}_2)L^+\psi_d({x}_2,{x}_1)+L^-\psi_d({x}_1,{x}_2)$ is a decomposition function for the system in \eqref{eq:eqiv_sys}, and equivalently, for System \eqref{eq:system}. Consequently, the $2n$-dimensional system $\begin{bmatrix}{(\overline{x}}_t^+)^\top & ({\underline{x}}_t^+)^\top \end{bmatrix}^\top=\begin{bmatrix} {f}^\top_d(\overline{x}_t,\underline{x}_t) & {f}^\top_d(\underline{x}_t,\overline{x}_t)\end{bmatrix}^\top$ {with initial condition $\begin{bmatrix} \overline{x}_0^\top & \underline{x}_0^\top\end{bmatrix}^\top$}, is an embedding system for \eqref{eq:system} (cf. Definition \ref{def:embedding}). So, $\underline{x}_t \leq x_t \leq \overline{x}_t$, by Proposition \ref{cor:embedding}.
\end{proof}
\subsection{Stability}
\yong{Besides the correctness property that we already obtain\mkr{ed} by construction,} 
we are interested in studying \yong{the} stability of the proposed 
\mk{framer}. In other words, we \yong{wish to} design the observer gain $L$ such that the observer error, $\varepsilon_t \triangleq \overline{x}_t-\underline{x}_t$, converges to zero asymptotically (cf. Definitions \ref{defn:error} and \ref{defn:stability}). Before stating our main results on observer stability, we derive some upper bounds for \yong{the interval widths of the JSS functions in terms of the interval widths of their domains}, 
which will be helpful in deriving the 
stability conditions.  
\begin{lem}[JSS Function \yong{Interval Width} Bounding] \label{lem:func_increment}
Let $f:\mathcal{X} \subset \mathbb{R}^n \to \mathbb{R}^p$ \mk{be a mapping that satisfies the assumptions in Proposition \ref{prop:JSS_decomp} and hence, can be decomposed in the form of \eqref{eq:JSS_decomp}}.
 Let $\mu_d \triangleq [\mu_{d,1}\dots \mu_{d,p}]^\top : \mathcal{X} \times \mathcal{X} \to \mathbb{R}^p$ be the tight decomposition function for the JSS mapping $\mu(\cdot)$, given in Proposition \ref{prop:tight_decomp}. \mk{Then,} for any \yong{interval domain $\underline{x}\le x\le\overline{x},$ with} 
$x, \underline{x},\overline{x} \in \mathcal{X}$, 
the following inequality 
holds:
\begin{align}\label{eq:increment_bounding}
\hspace{-.1cm}\Delta \mu_{d} \leq \overline{F}_{\mu} \yong{\varepsilon}, \ \hspace{-.1cm}\text{where}\hspace{-.1cm} \ \overline{F}_{\mu} \triangleq \hspace{-.1cm} 2\max(\overline{J}_f\hspace{-.1cm}-\hspace{-.1cm}H,\mathbf{0}_{p.n})\hspace{-.1cm}-\hspace{-.1cm}\underline{J}_f\hspace{-.1cm}+\hspace{-.1cm}H.
\end{align} 
\mk{where $\Delta \mu_d \triangleq [\Delta \mu_{d,1} \dots \Delta \mu_{d,p}]^\top \triangleq \mu_d(\overline{x},\underline{x})-\mu_d(\underline{x},\overline{x})$ and $\varepsilon$ is the interval width of $[\underline{x},\overline{x}]$ (cf. Definition \ref{defn:interval})}.
\end{lem}
\begin{proof}
First, note that using Proposition \ref{prop:tight_decomp}, $\forall i \in \mathbb{N}_m, \Delta \mu_{d,i} \triangleq \mu_{d,i}(\overline{x},\underline{x})-\mu_{d,i}(\underline{x},\overline{x})=\mu_i(D_i\overline{x}+(I-D_i)\underline{x})-\mu_i(D_i\underline{x}+(I-D_i)\overline{x})$, where $D_i$ is given in \eqref{eq:Dj} \yong{with $\overline{J}_{\mu,i}$ defined below}. Applying the mean value theorem, the last equality can be rewritten as $\Delta \mu_{d,i}=J_{\mu,i}(\xi)(D_i\overline{x}+(I-D_i)\underline{x}-D_i\underline{x}-(I-D_i)\overline{x}))=J_{\mu,i}(\xi)((2D_i-I)\overline{x}-(2D_i-I)\underline{x})=J_{\mu,i}(\xi)(2D_i-I)\yong{\varepsilon}$, where $\xi \in [\underline{x},\overline{x}]$, $J_{\mu,i}(\xi) \in [\underline{J}_{\mu,i},\overline{J}_{\mu,i}]=[\underline{J}_{f,i}-H_i,\overline{J}_{f,i}-H_i]$ and $H_i$ is the $i$'th row of $H$. \yong{Combining} this, Proposition \ref{prop:bounding} and the facts that $2D_i-I \geq 0$ (since $D_i$ is binary diagonal) and $\yong{\varepsilon} \geq 0$ (by \yong{the} correctness \yong{property}) \yong{yields} 
\begin{align}\label{eq:F_j}
\hspace{-.1cm}\Delta \mu_{d,i}\hspace{-.1cm} \leq\hspace{-.1cm} \overline{F}_{\mu,i} \yong{\varepsilon}, \text{where} \hspace{-.1cm} \ \overline{F}_{\mu,j} \hspace{-.1cm}\triangleq \hspace{-.1cm} 2\overline{J}_{f,i}D_i\hspace{-.1cm}-\hspace{-.1cm}\underline{J}_{f,i}\hspace{-.1cm}+\hspace{-.1cm}H_i(I\hspace{-.1cm}-\hspace{-.1cm}D_i).
\end{align} 
On the other hand, from \eqref{eq:Dj} we have $\overline{J}_{f,i}D_i=$
\begin{align}\label{eq:F_j_2}
\nonumber &[\overline{J}_{\mu,i,1}\max(\textstyle{\mathrm{sgn}}(J_{\mu,i,1}),{0}) \dots \overline{J}_{\mu,i,n}\max(\textstyle{\mathrm{sgn}}(J_{\mu,i,n}),{0})]\\
&\quad \quad \quad  =\max(\overline{J}_{\mu,i},\mathbf{0}_{1,n})=\max(\overline{J}_{f,i}-H_i,\mathbf{0}_{1,n}),
\end{align}
 where the second equality can be verified element-wise: for $j \in \mathbb{N}_n$, if $J_{\mu,i,j}$ is positive sign-stable, then $\textstyle{\mathrm{sgn}}(J_{\mu,i,j})=1$, hence $\max(\textstyle{\mathrm{sgn}}(J_{\mu,i,j}),{0})=1$, and therefore $\overline{J}_{\mu,i,j}\max(\textstyle{\mathrm{sgn}}(J_{\mu,i,j}),{0})=\overline{J}_{\mu,i,j}$. Moreover, the $j$'th entry of the row vector $\max(\overline{J}_{\mu,i},\mathbf{0}_{1,n})$ also equals $\max(\overline{J}_{\mu,i,j},{0})=\overline{J}_{\mu,i,j}$ since $\overline{J}_{\mu,i,j} \geq 0$ by positive sign-stability. The verification process can be done for the negative sign-stable case through similar reasoning. 
 Next, combining \eqref{eq:F_j} and \eqref{eq:F_j_2} 
 results in $\overline{F}_{\mu,i}=2\max(\overline{J}_{f,i}-H_i,\mathbf{0}_{1,n})-\underline{J}_{f,i}+H_i$, which when plugged into \eqref{eq:F_j} and stacked for all $i \in \mathbb{N}_p$, returns the result in \eqref{eq:increment_bounding}.     
\end{proof} 
Now, \mk{equipped} with the tools in Lemma \ref{lem:func_increment}, we derive sufficient LMI's to synthesize the \yong{stabilizing} observer gain $L$ for both DT and CT systems through the following theorem.
\begin{thm}[Stability]\label{thm:stability}
Consider the nonlinear plant $\mathcal{G}$ in \eqref{eq:system} and suppose all the assumptions in Lemma \ref{lem:correctness} hold. Then, the proposed correct interval \mk{framer} $\hat{\mathcal{G}}$ is stable, \mk{and hence, is an interval observer} in the sense of Definition \ref{defn:stability}, if there exist matrices $\mathbb{R}^{n \times n} \ni P \succ \mathbf{0}_{n,n}, X \in \mathbb{R}^{n \times n}$ and $J \in \mathbb{R}^{l \times n}$, \yong{$J \le 0$}, such that
\begin{enumerate}[(i)]
\item 
(if $\mathcal{G}$ is a CT system)  
\begin{align}\label{eq:CT_stability}
\hspace{-0.2cm}\begin{bmatrix} \Omega & \Lambda \\ \Lambda^\top & -\alpha (X+X^\top) \end{bmatrix}  \prec   0, 
\yong{J^\top C} \in  \mathbb{M}_n, X \in \mathbb{D}_n,
\end{align}
\yong{for all $\alpha >0$,} where 
\yong{$\Omega \triangleq ((A^m)^\top+\overline{F}_{\phi}^\top) X+X^\top (A^m+\overline{F}_{\phi})+(C^\top-\overline{F}_{\psi}^\top) J +J^\top (C  -\top \overline{F}_{\psi}) $ and $\Lambda \triangleq P+\alpha ((A^m)^\top+\overline{F}_{\phi}^\top) X+\alpha (C^\top  - \overline{F}_{\psi}^\top) J$;} 
\item (if $\mathcal{G}$ is a DT system) 
\begin{align}\label{eq:DT_stability}
\hspace{-.3cm}\begin{bmatrix} -P & \Gamma \\ \Gamma^\top & P- X- X^\top \end{bmatrix} \prec  0, 
J^\top C \leq 0, \yong{-X  \in \mk{\mathbb{M}_n}},
\end{align}
where $\Gamma \triangleq (|A|^\top+\overline{F}_{\phi}^\top)X-(|C|^\top+\overline{F}_{\psi}^\top)J$. 
\end{enumerate}
Furthermore, in both cases, $\overline{F}_{\phi}$ and $\overline{F}_{\psi}$ are computed by applying Lemma \ref{lem:func_increment} on the JSS functions $\phi$ and $\psi$, respectively. Finally, the corresponding stabilizing observer gain $L$ can be obtained as $L=-(X^\top)^{-1}J^\top .$
\end{thm}
\begin{proof}
Starting from \eqref{eq:observer}, we first derive the \mk{framer} error ($\varepsilon_t \triangleq \overline{x}_t- \underline{x}_t$) dynamics. Then, we show that the provided conditions in \eqref{eq:CT_stability} and \eqref{eq:DT_stability} are sufficient for stability of the error system in the CT and DT cases, respectively. To do so, define $\Delta \mu_d  \triangleq \mu_d(\overline{x},\underline{x})-\mu_d(\overline{x},\underline{x}), \forall \mu \in \{\phi,\psi \}$ and note that the LMIs in \eqref{eq:CT_stability} and \eqref{eq:DT_stability} and Schur complements 
imply that $X$ is positive definite and hence invertible (non-singular) in both CT and DT cases. 

Now, considering the CT case, from \eqref{eq:observer} and \eqref{eq:T_CT}, we obtain  \yong{the observer error dynamics:}
\begin{align}
\label{eq:error_dynamics}\dot{\varepsilon}_t&=((A\hspace{-.1cm}-\hspace{-.1cm}LC)^\text{d}+|(A\hspace{-.1cm}-\hspace{-.1cm}LC)^\text{nd}|)\varepsilon_t+\Delta \phi_d\hspace{-.1cm}+\hspace{-.1cm}|L||\Delta \psi_d| \\
\nonumber & \leq (A^\text{d}-(LC)^\text{d}+|A^\text{nd}|\hspace{-.1cm}+\hspace{-.1cm}|-(LC)^\text{nd}|+\overline{F}_{\phi}+|L|\overline{F}_{\psi})\varepsilon_t \\
&=(A^\text{m}+(-LC)^\text{m}+\overline{F}_{\phi}+|L|\overline{F}_{\psi})\varepsilon_t,\label{eq:error_dynamics_2}
\end{align} 
where $\forall \mu \in \{\phi,\psi \},\overline{F}_{\mu}$ is given in \eqref{eq:increment_bounding}, the inequality holds by Lemma \ref{lem:func_increment}, Proposition \ref{prop:bounding}, and the facts that $\forall M,N \in \mathbb{R}^{n \times n}$, $(M+N)^\text{d}=M^\text{d}+N^\text{d},(M+N)^\text{nd}=M^\text{nd}+N^\text{nd}$, $|M+N| \leq |M|+|N|$ by triangle inequality and the fact that $\varepsilon_t \geq 0$ by the correctness \yong{property} (Lemma \ref{lem:correctness}). Now, note that by the \emph{Comparison Lemma} \cite[Lemma 3.4]{khalil2002nonlinear} and positivity of the systems in \eqref{eq:error_dynamics} and \eqref{eq:error_dynamics_2}, stability of the system in \eqref{eq:error_dynamics_2} implies 
\yong{stability} for the actual error system in \eqref{eq:error_dynamics}. To show the former, 
we \yong{require the following:} 

i) 
$J$ and $X$ 
\yong{are} non-positive and diagonal matrices, respectively: This forces $X$ and its inverse to be diagonal matrices with strictly positive diagonal elements, and since $J$ is forced to be non-positive, $L=-(X^\top)^{-1}J^\top$ must be non-negative, and hence $|L|=L$; 

ii)  $J^\top C$ is Metzler: \yong{This results in $-LC=(X^\top)^{-1}J^\top C$ being Metzler, since} 
\mk{it} 
is a product of 
a diagonal and positive matrix $(X^\top)^{-1}$ and 
a Metzler matrix $J^\top C$, \yong{and it can be shown that their product is Metzler. Thus, $(-LC)^\text{m}=-LC$.}

By i) and ii), the system in \eqref{eq:error_dynamics_2} turns into the linear \yong{comparison system $\dot{\varepsilon}_t  \le  (A^\text{m}-LC+\overline{F}_{\phi}+L\overline{F}_{\psi})\varepsilon_t$}, whose stability is guaranteed by the LMI in \eqref{eq:CT_stability} by \cite[(12)]{pipeleers2009extended}. 

For the DT case, from \eqref{eq:observer} and \eqref{eq:T_DT} and by similar reasoning to the CT case, we obtain  
\begin{align}
\label{eq:error_dynamics_DT}{\varepsilon}_{t+1}&=|A-LC|\varepsilon_t+\Delta \phi_d+|L||\Delta \psi_d| \\
&\leq(|A|+|LC|+\overline{F}_{\phi}+|L|\overline{F}_{\psi})\varepsilon_t. \label{eq:error_dynamics_2_DT}
\end{align} 
  In addition, we enforce \yong{$-X$ to be Metzler,} 
  as well as $J$ and $J^\top C$ to be non-positive. 
  Consequently, \yong{since $X$ is positive definite,} 
  $X$ becomes a non-singular M-matrix\footnote{An M-matrix is a square matrix \yong{whose negation is Metzler and whose eigenvalues have nonnegative real parts}.}, and hence is inverse-positive \cite[Theorem 1]{plemmons1977m}, i.e., $X^{-1} \geq 0$. Therefore, $L=-(X^\top)^{-1}J^\top \geq 0$ and $LC=(X^\top)^{-1}(-J^\top C) \geq 0$, because they are matrix products of non-negative matrices, $(X^\top)^{-1},(-J^\top)$ and $(X^\top)^{-1},(-J^\top C)$, respectively. Hence, $|L|=L,|LC|=LC$, and so, the system in \eqref{eq:error_dynamics_2_DT} turns into \yong{${\varepsilon}_{t+1}\le(|A|+LC+\overline{F}_{\phi}+L\overline{F}_{\psi})\varepsilon_t$}, which is stable if the LMI in \eqref{eq:DT_stability} holds, by \cite[(10)]{pipeleers2009extended}.       
\end{proof}

\vspace{-0.1cm}
\mk{Finally, note that 
a coordinate transformation (cf. \cite{mazenc2021when} and references therein) may also be helpful for making the LMIs in Theorem \ref{thm:stability} feasible, as observed in Section \ref{sec:CT_exm}.}

\vspace{-0.1cm}
\section{Illustrative Examples}
The effectiveness of our interval observer design is illustrated for CT and DT systems \mk{(using YALMIP \cite{Lofberg2004}).}

\vspace{-0.1cm}
\subsection{CT System Example}\label{sec:CT_exm}
Consider 
the CT system\footnote{\mk{$a_1=35.63, b_1=15, a_2=0.25,a_3=36,a_4=200, \mathcal{X}_0 = [19.5, 9] \times [9, 11] \times [0.5,1.5]$.}} 
in \mk{\cite[Section IV, Eq. (30)]{dinh2014interval}}:
\mk{
\begin{gather*}
\dot{x}_{1} = x_{2}, \quad \dot{x}_{2}=b_1x_3-a_1\sin(x_1)-a_2x_2, \\ 
 \dot{x}_3=-a_2a_3x_1\hspace{-.05cm}+\hspace{-.05cm}\frac{a_1}{b_1}(a_4\sin(x_1)\hspace{-.05cm}+\hspace{-.05cm}\cos(x_1)x_2)\hspace{-.05cm}-\hspace{-.05cm}a_3x_2\hspace{-.05cm}-\hspace{-.05cm}a_4x_3,
\end{gather*}
with 
output $y=x_1$. 
}
\mk{Without a coordinate transformation, the LMIs in \eqref{eq:CT_stability} were infeasible, but with a coordinate transformation $z=Tx$ with $T={\footnotesize{\begin{bmatrix}20 & 0.1 & 0.1\\ 0 & 0.01 & 0.06\\0 &-10 & -0.4\end{bmatrix}}}$ and adding and subtracting $5y$ to the dynamics of $\dot{x}_1$, we}  
obtained the 
observer gain \mk{$L\hspace{-0.05cm}=\hspace{-0.05cm}10^{-6}\times [3.44 \ \, 0\ \, 0.04]^\top\hspace{-0.05cm}$}. 
As shown in 
Figure \ref{fig:figure1} \mk{($x_1,x_2$ omitted for brevity)}, 
 the state framers \mk{returned by our approach, 
 $\underline{x},\overline{x}$, are tighter than the ones obtained by the interval observer in \cite{dinh2014interval}, $\underline{x}^{DMN},\overline{x}^{DMN}$ (primarily because of outer-approximations of the initial framers $\mathcal{X}_0$ due to different coordinate transformations). Further, 
the framer error $\varepsilon_t= \overline{x}_t-\mkr{\underline{x}}_t$ is observed to 
tend to zero asymptotically.} 

 \subsection{DT System Example}\label{sec:DT_exm}
Consider a variant of DT H\'enon chaos system model \cite{Observer_discrete}:
\begin{align}
\label{eq:exampletwo}   
x_{t+1} =  Ax_t + r[1 -x_{t,1}^2 ], \quad y_t =  x_{t,1},
\end{align}
where 
$A =
\small{\begin{bmatrix}
0 & 1\\
0.3 & 0 
\end{bmatrix}}$, $r =\small{\begin{bmatrix}0.05 \\ 0\end{bmatrix}}$ and $\mathcal{X}_0 = [-2, 2] \times [-1, 1]$. 
Employing YALMIP to solve the corresponding LMIs in \eqref{eq:DT_stability}, a stabilizing observer gain $L =[ 0.0393 \ 0.0346]^\top$ was obtained.  
Figure \ref{fig:figure2} \mk{($x_1$ omitted for brevity)}
 \mk{shows that our observer returns comparable estimates to those obtained from the approach in \cite{tahir2021synthesis} and the framer error converges to zero.} 
 

\begin{figure}[t!] 
\centering
{\includegraphics[width=0.48\columnwidth,trim=5mm -5mm 10mm 0mm]{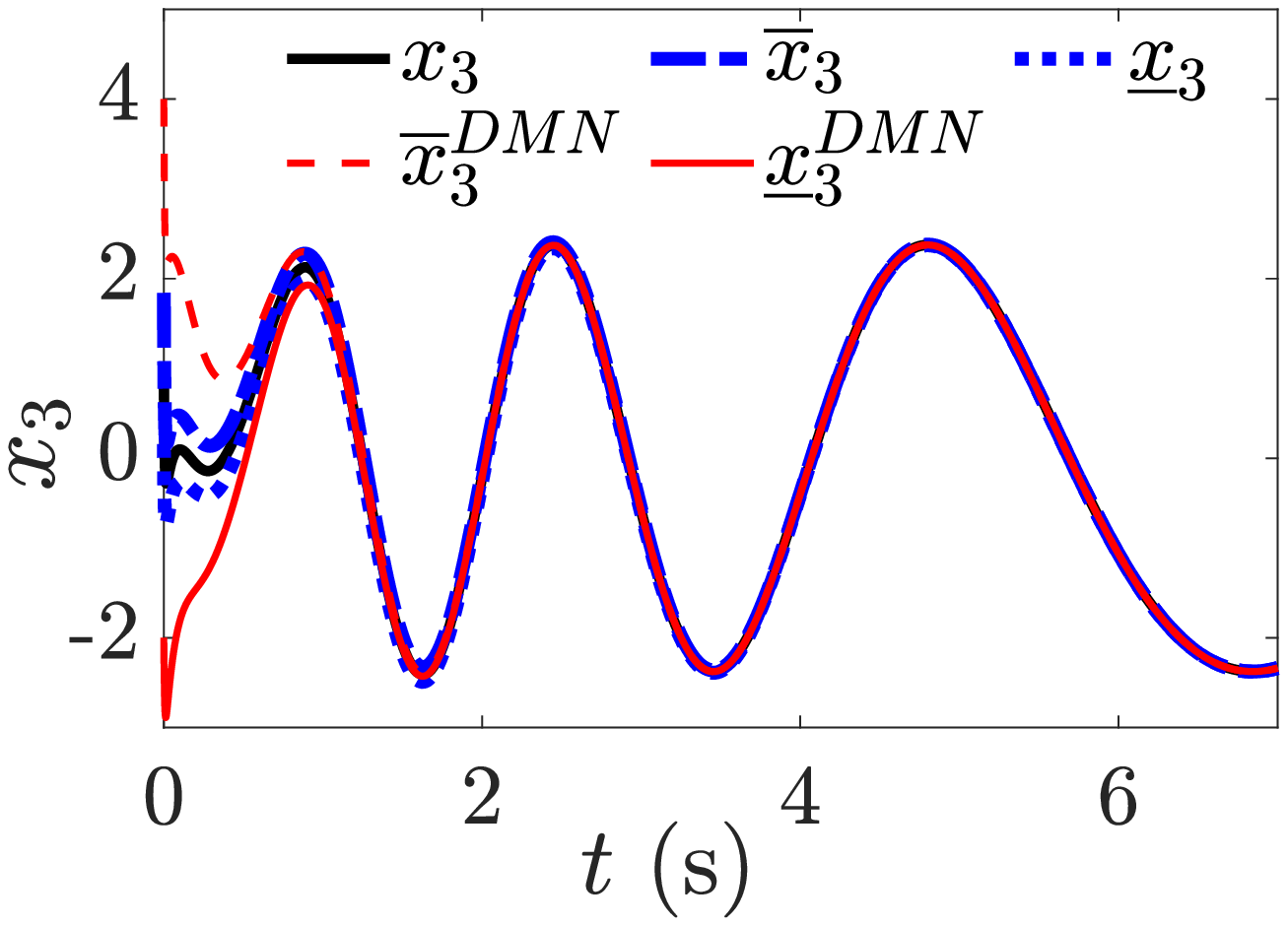}} \label{fig:sub1}
{\includegraphics[width=0.48\columnwidth,trim=5mm -5mm 10mm 0mm]{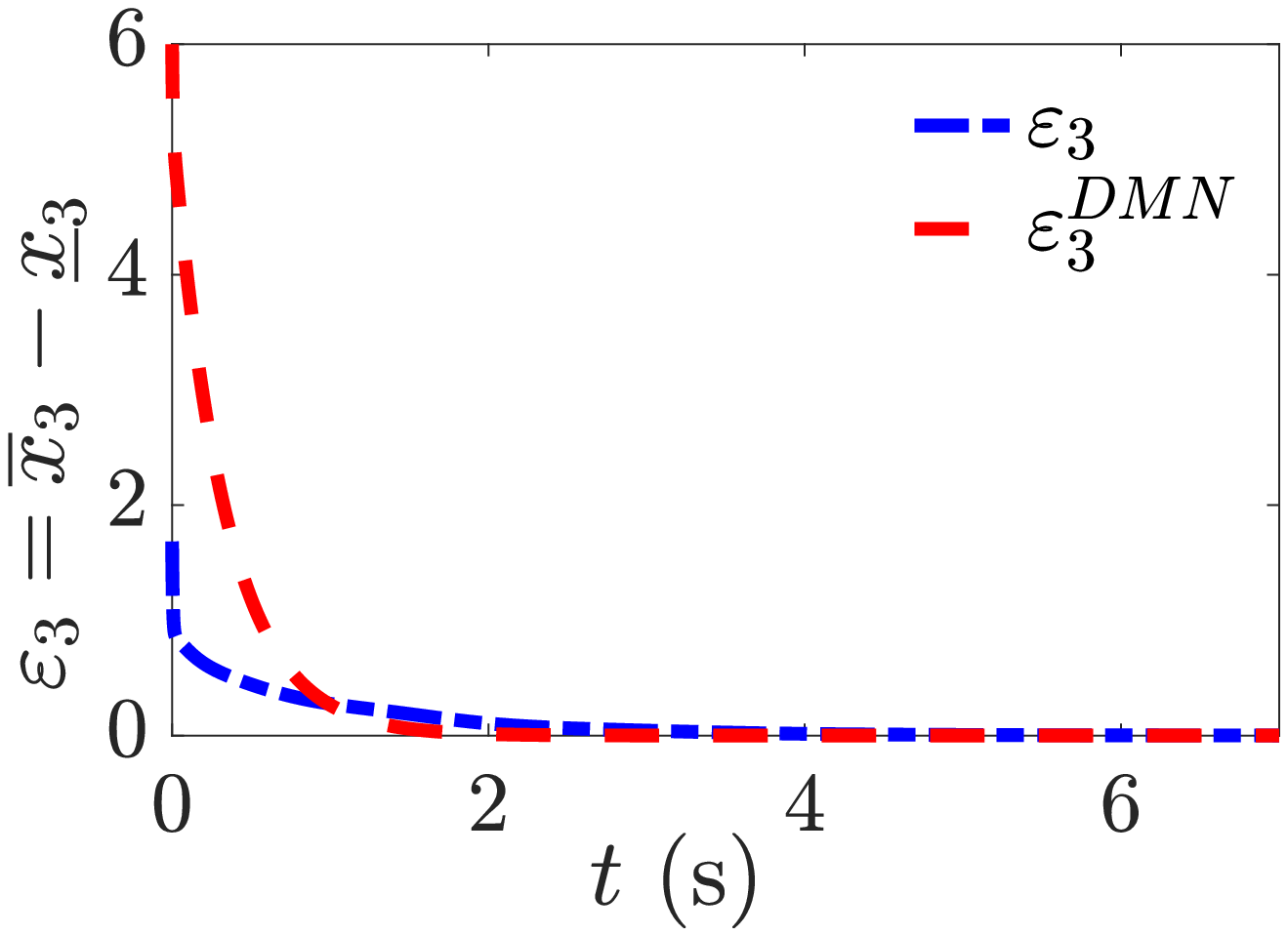}} \label{fig:sub2}
\vspace{-.3cm}
\caption{{{\small State, $x_3$, as well as its upper and lower framers and error \mk{returned by our proposed observer}, $\overline{x}_3,\underline{x}_3,\varepsilon_3$, \mk{and by the observer in \cite{dinh2014interval}, $\overline{x}^{DMN}_3,\underline{x}^{DMN}_3,\varepsilon^{DMN}_3$} for the CT System example.}}}
\label{fig:figure1}
\end{figure}

\begin{figure}[t!] 
\centering
{\includegraphics[width=0.48\columnwidth,trim=10mm 5mm 10mm 0mm]{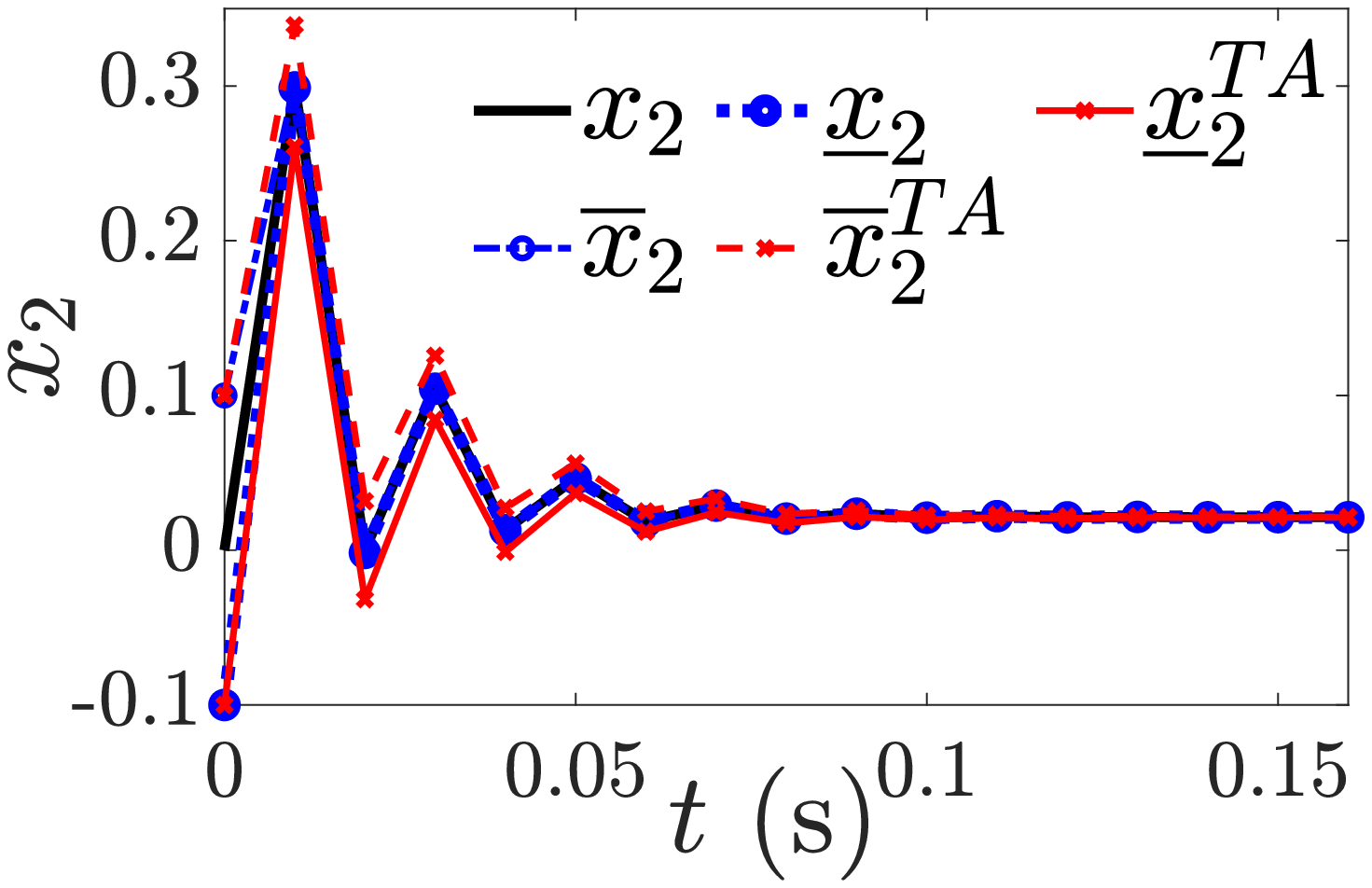}}\label{fig:sub4}
{\includegraphics[width=0.48\columnwidth,trim=7mm 5mm 10mm 0mm]{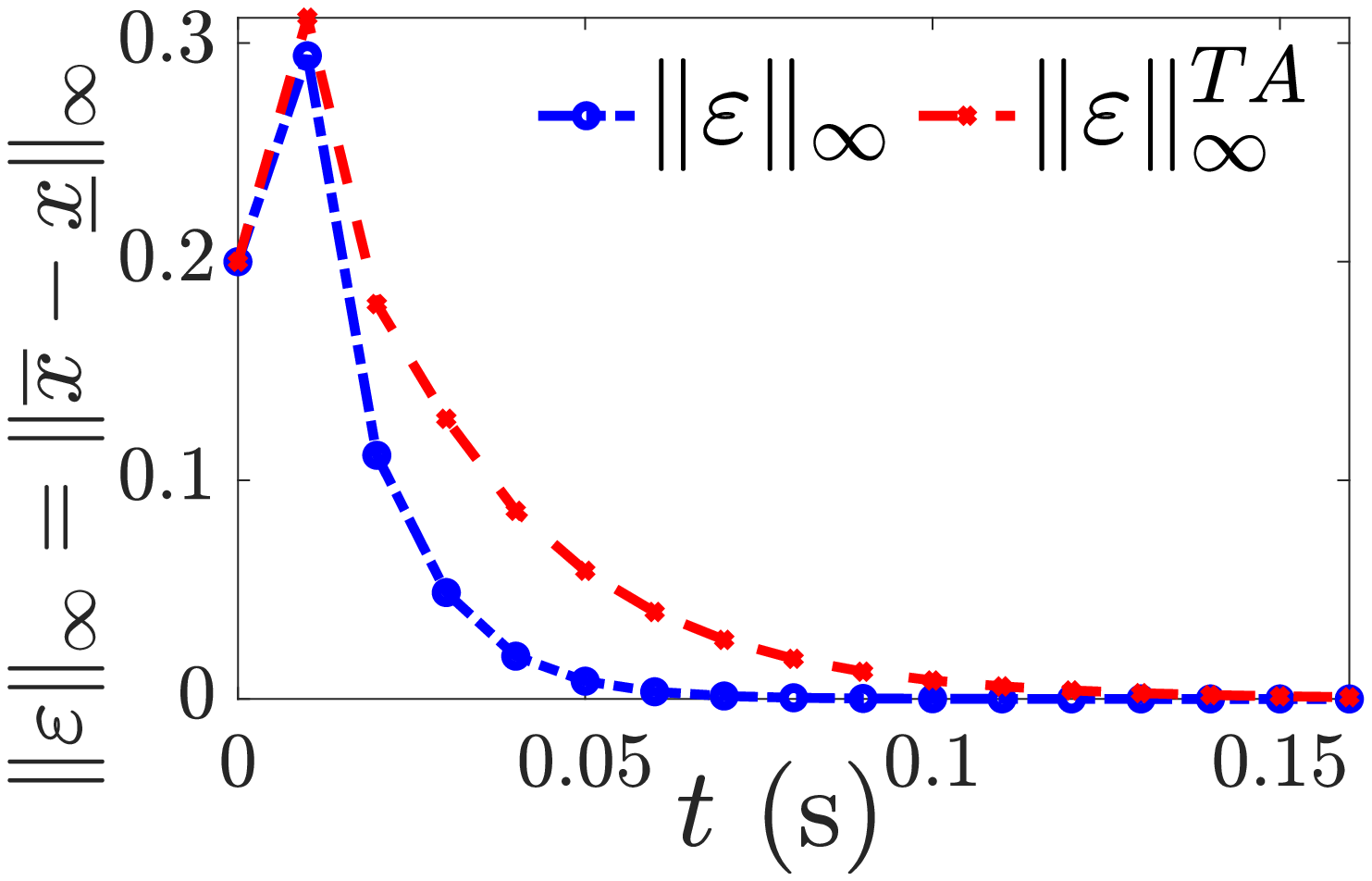}}\label{fig:sub5}
\caption{\small\mk{State, $x_2$, and its upper and lower framers, \mk{returned by our proposed observer}, $\overline{x}_2,\underline{x}_2$, \mk{and by the observer in \cite{tahir2021synthesis}, $\overline{x}^{TA}_2,\underline{x}^{TA}_2$} (left) and norm of framer error (right) for the DT System example.}}
\label{fig:figure2}
\end{figure}

\section{Conclusion} \label{sec:conclusion}

A novel unified interval observer synthesis approach was presented for locally Lipschitz nonlinear \yong{continuous-time (CT) and discrete-time (DT)} systems with nonlinear observations. Leveraging {mixed-monotone decompositions}, the proposed observer satisfies the \emph{\yong{correctness} property} by construction, i.e., the true state trajectory of the system was shown to be \emph{framed} by the states of the observer \yong{at} all times, \yong{without needing restrictive assumptions such as global Lipschitz continuity or contraction. 
Moreover, by solving a semi-definite program based on some sufficient conditions with LMIs, a stabilizing observer gain was designed to ensure that the observer errors converge to zero asymptotically.} Finally, the effectiveness of the proposed observer, \mk{
when compared to some benchmark observers, was} demonstrated using illustrative  DT and CT system examples. \mo{In our future work, we will consider noise and uncertainties within our framework as well as hybrid system dynamics, and extend 
our approach to simultaneous input and state \yong{estimation}. } 
\bibliographystyle{unsrturl}

\bibliography{biblio}

\begin{thebibliography}{10}

\bibitem{wang2015interval}
Y.~Wang, D.M. Bevly, and R.~Rajamani.
\newblock Interval observer design for {LPV} systems with parametric
  uncertainty.
\newblock {\em Automatica}, 60:79--85, 2015.

\bibitem{mazenc2013robust}
F.~Mazenc, T-N. Dinh, and S-I. Niculescu.
\newblock Robust interval observers and stabilization design for discrete-time
  systems with input and output.
\newblock {\em Automatica}, 49(11):3490--3497, 2013.

\bibitem{mazenc2011interval}
F.~Mazenc and O.~Bernard.
\newblock Interval observers for linear time-invariant systems with
  disturbances.
\newblock {\em Automatica}, 47(1):140--147, 2011.

\bibitem{chebotarev2015interval}
S.~Chebotarev, D.~Efimov, T.~Ra{\"\i}ssi, and A.~Zolghadri.
\newblock Interval observers for continuous-time {LPV} systems with {L1/L2}
  performance.
\newblock {\em Automatica}, 58:82--89, 2015.

\bibitem{moisan2007near}
M.~Moisan, O.~Bernard, and J-L. Gouz{\'e}.
\newblock Near optimal interval observers bundle for uncertain bioreactors.
\newblock In {\em European Control Conference (ECC)}, pages 5115--5122. IEEE,
  2007.

\bibitem{raissi2010interval}
T.~Ra{\"\i}ssi, G.~Videau, and A.~Zolghadri.
\newblock Interval observer design for consistency checks of nonlinear
  continuous-time systems.
\newblock {\em Automatica}, 46(3):518--527, 2010.

\bibitem{tahir2021synthesis}
A.M. Tahir and B.~A{\c{c}}{\i}kme{\c{s}}e.
\newblock Synthesis of interval observers for bounded {J}acobian nonlinear
  discrete-time systems.
\newblock {\em IEEE Control Systems Letters}, 2021.

\bibitem{khajenejad2021intervalACC}
M.~Khajenejad, Z.~Jin, and S.Z. Yong.
\newblock Interval observers for simultaneous state and model estimation of
  partially known nonlinear systems.
\newblock In {\em American Control Conference (ACC)}, pages 2848--2854, 2021.

\bibitem{khajenejad2020simultaneousCDC}
M.~Khajenejad and S.Z. Yong.
\newblock Simultaneous input and state interval observers for nonlinear systems
  with full-rank direct feedthrough.
\newblock In {\em IEEE Conference on Decision and Control}, pages 5443--5448,
  2020.

\bibitem{farina2000positive}
L.~Farina and S.~Rinaldi.
\newblock {\em Positive linear systems: theory and applications}, volume~50.
\newblock John Wiley \& Sons, 2000.

\bibitem{chambon2016overview}
E.~Chambon, L.~Burlion, and P.~Apkarian.
\newblock Overview of linear time-invariant interval observer design: towards a
  non-smooth optimisation-based approach.
\newblock {\em IET Control Theory \& Applications}, 10(11):1258--1268, 2016.

\bibitem{cacace2014new}
F.~Cacace, A.~Germani, and C.~Manes.
\newblock A new approach to design interval observers for linear systems.
\newblock {\em IEEE Transactions on Automatic Control}, 60(6):1665--1670, 2014.

\bibitem{kieffer2006guaranteed}
M.~Kieffer and E.~Walter.
\newblock Guaranteed nonlinear state estimation for continuous-time dynamical
  models from discrete-time measurements.
\newblock {\em IFAC Proceedings Volumes}, 39(9):685--690, 2006.

\bibitem{efimov2013interval}
D.~Efimov, T.~Ra{\"\i}ssi, S.~Chebotarev, and A.~Zolghadri.
\newblock Interval state observer for nonlinear time varying systems.
\newblock {\em Automatica}, 49(1):200--205, 2013.

\bibitem{abate2020tight}
M.~Abate, M.~Dutreix, and S.~Coogan.
\newblock Tight decomposition functions for continuous-time mixed-monotone
  systems with disturbances.
\newblock {\em IEEE Control Systems Letters}, 5(1):139--144, 2020.

\bibitem{yang2019sufficient}
L.~Yang, O.~Mickelin, and N.~Ozay.
\newblock On sufficient conditions for mixed monotonicity.
\newblock {\em IEEE Transactions on Automatic Control}, 64(12):5080--5085,
  2019.

\bibitem{khajenejad2021tight}
M.~Khajenejad and S.Z. Yong.
\newblock Tight remainder-form decomposition functions with applications to
  constrained reachability and interval observer design.
\newblock {\em arXiv preprint arXiv:2103.08638}, 2021.

\bibitem{khajenejad2021simultaneousECC}
M.~Khajenejad and S.Z. Yong.
\newblock Simultaneous input and state interval observers for nonlinear systems
  with rank-deficient direct feedthrough.
\newblock In {\em European Control Conference (ECC)}. IEEE, 2021.

\bibitem{khalil2002nonlinear}
H.K. Khalil.
\newblock Nonlinear systems.
\newblock {\em Upper Saddle River}, 2002.

\bibitem{pipeleers2009extended}
G.~Pipeleers, B.~Demeulenaere, J.~Swevers, and L.~Vandenberghe.
\newblock Extended {LMI} characterizations for stability and performance of
  linear systems.
\newblock {\em Systems \& Control Letters}, 58(7):510--518, 2009.

\bibitem{plemmons1977m}
R.J. Plemmons.
\newblock {M}-matrix characterizations. {I:} nonsingular {M}-matrices.
\newblock {\em Linear Algebra and its Applications}, 18(2):175--188, 1977.

\bibitem{mazenc2021when}
Frederic Mazenc and Olivier Bernard.
\newblock When is a matrix of dimension 3 similar to a metzler matrix
  application to interval observer design.
\newblock {\em IEEE Transactions on Automatic Control}, pages 1--1, 2021.

\bibitem{Lofberg2004}
J.~L{\"{o}}fberg.
\newblock Yalmip: A toolbox for modeling and optimization in matlab.
\newblock In {\em CACSD Conference}, Taipei, Taiwan, 2004.

\bibitem{dinh2014interval}
T.N. Dinh, F.~Mazenc, and S.~Niculescu.
\newblock Interval observer composed of observers for nonlinear systems.
\newblock In {\em European Control Conference (ECC)}, pages 660--665. IEEE,
  2014.

\bibitem{Observer_discrete}
D.~Efimov, W.~Perruquetti, T.~Ra{\"\i}ssi, and A.~Zolghadri.
\newblock On interval observer design for time-invariant discrete-time systems.
\newblock In {\em European Control Conference (ECC)}. IEEE, 2013.

\end{thebibliography}

\end{document}